\documentclass[twocolumn]{IEEEtran} 

\usepackage{amsmath,amssymb,amsfonts}
\usepackage[final]{graphicx}
\usepackage{psfrag}
\usepackage[tight]{subfigure}
\usepackage[numbers,sort&compress]{natbib}

\usepackage{mathrsfs} 
\usepackage{epstopdf}
\usepackage{multirow}
\usepackage{amsthm,mathtools}
\usepackage{float}
\usepackage{booktabs}  
\usepackage{multirow} 
\usepackage{relsize}
\usepackage{bm}
\usepackage{hyperref}
\usepackage[nolist,nohyperlinks]{acronym}
\usepackage{svg}
\usepackage{xparse}
\NewDocumentCommand{\evalat}{sO{\big}mm}{%
  \IfBooleanTF{#1}
   {\mleft. #3 \mright|_{#4}}
   {#3#2|_{#4}}%
}
\newcommand{\norm}[1]{\left\lVert#1\right\rVert}

\begin{acronym}
\acro{PLS}{physical layer security}
\acro{SNR}{signal-to-noise ratio}
\acrodefplural{SNR}[SNRs]{signal-to-noise ratios}
\acro{AWGN}{additive white Gaussian noise}
\acro{CSI}{channel state information}
\acro{PDF}{probability density function}
\acrodefplural{PDF}[PDFs]{probability density functions}
\acro{CDF}{cumulative distribution function}
\acrodefplural{CDF}[CDFs]{cumulative distribution functions}
\acro{LRS}{large reflecting surface}
\acro{BS}{base station}
\acro{RV}{random variable}
\acro{FN}{folded normal}
\acro{TDD}{time-division duplexing}
\acro{LOS}{line-of-sight}
\acro{DL}{downlink}
\acro{MRT}{maximal ratio transmission}
\acro{MRC}{maximal ratio combining}
\acro{MC}{Monte Carlo}
\end{acronym}

\newtheorem{proposition}{Proposition}

\newcommand{\diag}{\mathop{\mathrm{diag}}}
\DeclareMathOperator{\Var}{Var}
\DeclareMathOperator{\tr}{tr}

\usepackage[numbers,sort&compress]{natbib}

\DeclareMathOperator{\sinc}{sinc}

\makeatletter
\def\blfootnote{\xdef\@thefnmark{}\@footnotetext}
\makeatother

\usepackage{float}
\usepackage{stfloats}
\usepackage{textcomp}
\usepackage[ruled,vlined]{algorithm2e}

\begin{document}
\title{\huge{ Physical Layer Security of RIS-Assisted \\Communications under Electromagnetic Interference}}
\author{Jos\'e~David~Vega-S\'anchez, Georges Kaddoum, and F. Javier L\'opez-Mart\'inez}
\maketitle
\blfootnote{\noindent Manuscript received MONTH xx, YEAR; revised XXX. The review of this paper was coordinated by XXXX. 
The work of F.J. Lopez-Martinez was funded by MCIN/AEI/10.13039/501100011033 through grant PID2020-118139RB-I00, and by Junta de Andalucia (P18-RT-3175).}
\blfootnote{\noindent J.~D.~Vega~S\'anchez is with Departamento de Electr\'onica, Telecomunicaciones y Redes de Informaci\'on, Escuela Polit\'ecnica Nacional (EPN),
Quito,  170525, Ecuador. (e-mail: $\rm jose.vega01@epn.edu.ec$).}
\blfootnote{\noindent Georges Kaddoum is with Department of Electrical Engineering, ETS, University of Quebec, Montreal, QC H3C 1K3, Canada. (e-mail: $\rm georges.kaddoum@etsmtl.ca$).}

\blfootnote{\noindent F.~J. L\'opez-Mart\'inez is with Communications and Signal Processing Lab, Telecommunication Research Institute (TELMA), Universidad de M\'alaga, E.T.S. Ingenier\'ia de Telecomunicaci\'on, Bulevar Louis Pasteur 35, 29010
M\'alaga (Spain). (e-mail: $\rm fjlopezm@ic.uma.es$).}

\vspace{-12.5mm}
\begin{abstract}
This work investigates the impact of the ever-present electromagnetic interference (EMI) on the achievable secrecy performance of reconfigurable intelligent surface (RIS)-aided communication systems. We characterize the end-to-end RIS channel by considering key practical aspects such as spatial correlation, transmit beamforming vector, phase-shift noise, the coexistence of direct and indirect channels, and the presence of strong/mild EMI on the receiver sides. We show that the effect of EMI on secrecy performance strongly depends on the ability of the eavesdropper to cancel such interference; this puts forth the potential of EMI-based attacks to degrade physical layer security in RIS-aided communications.  
\end{abstract}


\vspace{-2.5mm}
\section{Introduction}
Reconfigurable intelligent surfaces (RIS) have drawn full attention thanks to their outstanding potential to enhance coverage, spectral/energy-efficiencies and security of forthcoming wireless systems. An RIS is usually built as a planar metasurface containing a large number of simple, nearly passive reflecting elements. These can be dynamically configured to tune the phases and amplitudes of the impinging waves on the RIS, helping to overcome the detrimental effects of the wireless channel \cite{Wu}. While RISs have a tremendous potential to be a game-changer technology on the verge of 6G, practical impairments related to their implementation and deployment are known to limit their performance; these are the cases of imperfect phase-shift compensation \cite{Badiu}, spatial correlation \cite{emilcorr} or, very recently, electromagnetic interference (EMI) \cite{andrea}.

One of the potential use cases of RIS is physical layer security (PLS) \cite{Yan2021}, as a way to remarkably improve network security by exploiting the inherent randomness (e.g., noise, fading, interference) \cite{RefPLS} of the wireless propagation medium. A large body of research has been motivated by the potential of integrating both RIS and PLS to enable secure and intelligent radio environments \cite{PLSRIS1,Yang2020,Zhang2021,PLSRIS2}. However, even though the effect of the EMI present in any wireless environment is known to affect RIS-aided communications \cite{andrea}, its effects in physical layer security have not been analyzed yet. Since EMI can appear from intentional and non-intentional (e.g., natural pollution/radiation) causes \cite{andrea}, it can play a key role in compromising the security of communications. From a PLS perspective, intentional signals can be induced in a cooperative jamming-eavesdropping fashion. In this situation, the eavesdropper may be able to cancel out such interference, which can provide an advantage over the legitimate node when decoding the information.

Based on the previous considerations, in this work we explore the PLS of RIS-assisted wireless communications affected by EMI. 
We consider a practical RIS scenario on which key aspects such as spatial correlation, phase-shift errors, and the coexistence of direct and indirect channels are considered, together with EMI. All these factors are integrated into statistical approximations for the equivalent signal-to-noise ratios (SNRs) of interest, which allow us to derive closed-form solutions for the secrecy outage probability (SOP) of the underlying system, assuming different considerations of EMI-awareness at the end nodes. Finally, through illustrative examples, we provide valuable insights on the role of EMI for the system's secure performance. 

\noindent\textit{Notation and terminology:} Upper and lower-case bold letters denote matrices and
vectors; $f_{(\cdot)}(\cdot)$ denotes a probability density function (PDF); $F_{(\cdot)}(\cdot)$ is a cumulative density function (CDF); $\mathcal{U}[a,b]$ denotes a uniform distribution on $[a,b]$; $\mathcal{C}\mathcal{N}(\cdot ,\cdot )$ is the circularly symmetric complex Gaussian distribution; $\mathbb{C}$ denotes the set of complex numbers;  $\mathbb{E}[\cdot]$ is expectation operator; $\Var\left [\cdot\right ]$ is variance; $\Gamma(\cdot)$ is the gamma function~\cite[Eq.~(6.1.1)]{Abramowitz}; $\Upsilon(\cdot,\cdot)$, the lower incomplete gamma function~\cite[Eq.~(6.5.2)]{Abramowitz}; ${}_2F_1\left(\cdot,\cdot;\cdot;\cdot\right)$ is Gauss hypergeometric function~\cite[Eq.~(15.1.1)]{Abramowitz}; $\mathcal{B}\left (\cdot,\cdot\right )$ is the Beta function~\cite[Eq.~(6.2.2)]{Abramowitz}; $\diag\left ( {\bf{x}} \right )$ is a diagonal matrix with diagonal given by ${\bf{x}}$; ${\bf{I}}_{{ N}}$ is the identity matrix of size $N\times N$; $\norm{\cdot}$ is the Euclidean norm of a complex vector; $\left ( \cdot \right )^{\rm H}$ is the  Hermitian transpose; $\rm{mod} \left ( \cdot \right )$ is the modulus operation; $\left \lfloor \cdot   \right \rfloor$ is the floor function, and $\sinc(w)=\sin(\pi w)/(\pi w)$ is the sinc function. 
\vspace{-2mm}
\section{System and Channel Models}
We consider an RIS-assisted wiretap system consisting of a transmitter node Alice ($\mathrm{A}$) with $M$ antennas communicating with a single-antenna legitimate receiver Bob ($\mathrm{B}$) via an RIS equipped with $N$ nearly passive reconfigurable elements in the presence of a single-antenna
eavesdropper Eve ($\mathrm{E}$), as illustrated in Fig.~\ref{sistema1}. We also assume that the RIS is subject to EMI, which is produced by controllable/uncontrollable sources in the far-field of the RIS. According to \cite{andrea}, the received signals at both $\mathrm{B}$ and $\mathrm{E}$ can be written as
\begin{equation}
\label{eq1}
y_i=\sqrt{P} \left ( {\bf{h}}_{2,i}^{\rm H}\bm{\Psi}{\bf{G}} + {\bf{h}}_{{\rm{d}},i}^{\rm H} \right ) {\bf{w}}x+ {\bf{h}}_{2,i}^{\rm H}\bm{\Psi}\bm{\nu}  +\widetilde{n}_i, 
\end{equation}
where $i \in \left \{ \mathrm{B},\mathrm{E} \right \}$ indicates either the legitimate or the eavesdropper channels, $P$ is the transmit power at $\mathrm{A}$, $x$ is the transmitted signal with $\mathbb{E}\{|x|^2\}=1$,
and $\widetilde{n}_i\sim \mathcal{C}\mathcal{N}(0,\sigma^2_{i})$ is the additive white Gaussian noise with $\sigma^2_{i}$ power. The EMI effect is denoted by the vector $\bm{\nu} \in \mathbb{C}^{N\times 1}$, ${\bf{h}}_{{\rm d},i} \in \mathbb{C}^{M\times 1}$ refers to the direct channel between $\mathrm{A}$-to-$\mathrm{B}$ or $\mathrm{A}$-to-$\mathrm{E}$,  $\bf{w}$ $\in \mathbb{C}^{M\times 1}$ denotes the beamforming vector at $\mathrm{A}$, ${\bf{G}}= \left [ \bf{g}_1,\dots, \bf{g}_{M} \right ] \in \mathbb{C}^{N\times M}$ and ${\bf{h}}_{2,i}=\left [ h_{2i,1},\dots,  h_{2i,N} \right ]^{\rm H} \in \mathbb{C}^{N\times 1}$ represent the channel coefficients for the paths $\mathrm{A}$-to-RIS and either RIS-to-$\mathrm{B}$ or RIS-to-$\mathrm{E}$, respectively.
Moreover, $\bm{\Psi}=\diag\left ( e^{j\phi_{1}},\dots, e^{j\phi_{N}}\right )$ indicates the phase-shift matrix induced by the RIS elements. For the sake of simplicity in the discussion, we assume that the RIS does not attenuate the reflected signals. 
Assuming that the legitimate agents $\mathrm{A}$ and $\mathrm{B}$ are not aware of the presence of $\mathrm{E}$, the RIS designs the phases-shifts of each reflecting element so that the signals arriving at $\mathrm{B}$ are aligned, i.e. $\phi_n=\angle( {\bf{h}}_{{\rm{d}},\mathrm{B}}^{\rm H}{\bf{w}})-\angle({ h}_{2,\mathrm{B},n}^{\rm H})-\angle({\bf{g}}_n{\bf{w}})$~\cite[Eq.~(19)]{MRT} where ${ h}_{2,\mathrm{B},n}^{\rm H}$ is the $n$th element of ${\bf{h}}_{2,\mathrm{B}}^{\rm H}$ and ${\bf{g}}_n$ is the $n$th row vector of ${\bf{G}}$. In such a setup, the phases $\phi_{1\ldots n}$ and the transmit beamforming vector $\bf{w}$ are jointly optimized \cite{MRT} to maximize the received signal-to-noise ratio (SNR) at $\mathrm{B}$. However, in practice, discrete phase-shifts and imperfect channel information at the RIS induces a residual random phase noise. We denote this error term as $\varphi_n$, so the designed phase-shifts of the $n$th RIS element deviate from the optimal ones as $\Delta_n=\phi_n+\varphi_n$ \cite{Badiu}.
 Based on the above, the composite channel observed by the receiver nodes without considering the EMI effect can be formulated as  \vspace{-2mm}
\begin{equation}
\label{eq2}
h_i=\left ( {\bf{h}}_{2,i}^{\rm H}\bm{\Phi}{\bf{G}} + {\bf{h}}_{{\rm{d}},i}^{\rm H} \right ) \bf{w}^{*}, 
\vspace{-1mm}
\end{equation}
\begin{figure}[t]
\centering 
\psfrag{A}[Bc][Bc][0.7]{Alice}
\psfrag{K}[Bc][Bc][0.7]{${\bf{G}}$}
\psfrag{Z}[Bc][Bc][0.7][-18.5]{${\bf{h}}_{\rm d,\mathrm{E}}$}
\psfrag{L}[Bc][Bc][0.7][-16]{${\bf{h}}_{\rm d,\mathrm{B}}$}
\psfrag{H}[Bc][Bc][0.7][70]{${\bf{h}}_{2,\mathrm{B}}$}
\psfrag{B}[Bc][Bc][0.7]{Bob}
\psfrag{E}[Bc][Bc][0.7]{Eve}
\psfrag{U}[Bc][Bc][0.7][80]{${\bf{h}}_{2,\mathrm{E}}$}
\psfrag{W}[Bc][Bc][0.7][-30]{RIS}
\psfrag{P}[Bc][Bc][0.7]{Electromagnetic}
\psfrag{Q}[Bc][Bc][0.7]{Interference}
\psfrag{R}[Bc][Bc][0.7]{$\bm{\Psi}$}
\psfrag{Main channel}[Bc][Bc][0.7]{Main channel} 
\psfrag{Wiretap channel}[Bc][Bc][0.7]{Wiretap channel} 
\includegraphics[width=0.7\linewidth]{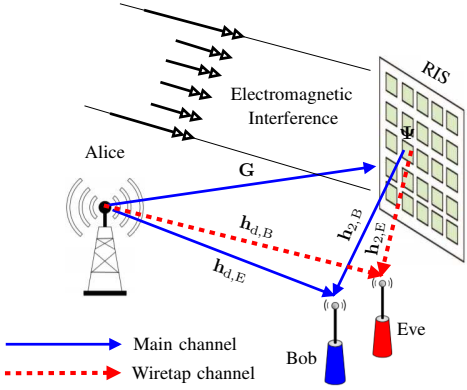} \caption{RIS-aided wiretap MISO system under EMI.}
\label{sistema1}
\end{figure}
where ${\bm{\Phi}}=\diag( e^{j\Delta_{1}},\dots, e^{j\Delta_{N}} )$ is the imperfect phase-shift matrix on the RIS and ${\bf{w}}^{*} =  \tfrac{\left ( {\bf{h}_{2,\mathrm{B}}^{\rm H}}{\bm{\Phi}}{\bf{G}} + {\bf{h}}_{{\rm{d}},\mathrm{B}}^{\rm H} \right )^{\rm H}}{\norm{ {\bf{h}_{2,\mathrm{B}}^{\rm H}}{\bm{\Phi}}{\bf{G}} + {\bf{h}}_{{\rm{d}},\mathrm{B}}^{\rm H} } }$ $ \in \mathbb{C}^{M\times 1} $ is the maximum ratio transmission (MRT) beamforming vector computed using the distributed algorithm described in~\cite{MRT}. Under isotropic scattering, the channel coefficients in \eqref{eq2} including spatial correlation on the RIS are formulated as \cite{emilcorr}
\vspace{-1.7mm}
\begin{align}\label{eq3}
{\bf{h}}_{{\rm d},i}&\sim \mathcal{C}\mathcal{N}\left (\bm{0}_{M}, \beta_{{\rm d},i}\vspace{0.8mm} {\bf{I}}_{ M}\right ), \hspace{2mm}  {\bf{h}}_{2,i}\sim \mathcal{C}\mathcal{N}\left ( \bm{0}_{ N},A \beta_{2,i} {\bf{R}} \vspace{0.8mm} {\bf{I}}_{ N}\right )  \nonumber \\ {\bf{g}}_{\rm q}&\sim \mathcal{C}\mathcal{N}\left ( \bm{0}_{N},A \beta_1 {\bf{R}} \vspace{0.8mm} {\bf{I}}_{ N}\right ) \ \text{for} \ q=\left \{ 1,\dots,M \right \},   
\end{align}
wherein $\beta_1$, $\beta_{2,i}$, and $\beta_{{\rm d},i}$ encompass the average path loss attenuation for the links between $\mathrm{A}$-to-RIS, RIS-to-$i$, and $\mathrm{A}$-to-$i$, respectively. Moreover, $A=d_{\rm H}d_{\rm V}$ is the area of a RIS element, where $d_{\rm H}$ is the horizontal width and $d_{\rm V}$ refers to the vertical height, and ${\bf{R}} \in \mathbb{C}^{N\times N}$ is the spatial correlation matrix on the RIS. To compute ${\bf{R}}$, we resort to the scheme in \cite{emilcorr} for a practical RIS setup. Hence, we assume a rectangular RIS geometry with $N=N_{\rm V} N_{\rm H}$ elements, where $N_{\rm V}$ and $N_{\rm H}$ indicate the number of elements per row and per column, respectively. In this context, and under isotropic Rayleigh fading, the $(a,b)$th entry of the spatial correlation matrix ${\bf{R}}$ can be expressed as 
\begin{equation}\label{eq4}
{\left [{\bf{R}}  \right ] }_{a,b}= \sinc\left ( 2\left \| {\bf{u}}_a-{\bf{u}}_b \right \| /\lambda \right )     \hspace{2mm}  a,b=1,\dots,N
\end{equation}
in which ${\bf{u}}_\zeta =\left [ 0, \rm{mod}\left ( \zeta -1,\mathrm{N}_H\right )d_{\rm H}, \left \lfloor \left (\zeta -1 \right )/N_{\rm H} \right \rfloor d_{\rm V} \right ]^{T}$, $\zeta  \in \left \{ a,b \right \}$, and $\lambda$ is the wavelength. Before getting into the received SNR formulations at the end nodes, the distribution of the EMI needs to be specified. Therefore, based on~\cite[Corollary 1]{andrea}, the EMI, denoted by $\bm{\nu}$, is distributed as
\begin{equation}\label{eq5}
\bm{\nu} \sim \mathcal{C}\mathcal{N}\left (\bm{0}_{ N}, A \sigma^2_{\mathrm{EMI}}\vspace{0.8mm} {\bf{R}} {\bf{I}}_{ N}\right ),
\end{equation}
where ${\bf{R}}$ under isotropic conditions is also given by \eqref{eq4}, and $\sigma^2_{\mathrm{EMI}}$ is the EMI power produced by incoming waves on the RIS, generated by external (either intentional or non-intentional) sources. With the above formulations and from \eqref{eq1}, the received SNR at $\mathrm{B}$ or $\mathrm{E}$ is given by \vspace{-2mm}
\begin{align}\label{eq6}
\gamma_{i}=&\frac{P | {{h_i}}|^2}{A \delta_i \sigma^2_{\mathrm{EMI}} {\bf{h}}_{2,i}^{\rm H}\bm{\Phi} {\bf{R}}{\bf{h}}_{2,i} +\sigma_{i}^2}=\overline{\gamma}_i X_i
\end{align}
where $X_i=\tfrac{| {{h_i}}|^2}{
    \frac{ A \delta_i \sigma^2_{\mathrm{EMI}}}{\sigma_{i}^2} {\bf{h}}_{2,i}^{\rm H}\bm{\Phi} {\bf{R}}{\bf{h}}_{2,i} +1}$ and for the sake of convenience, we define $\overline{\gamma}_i=P/\sigma_{i}^2$ as the average transmit SNRs for $\mathrm{B}$ or $\mathrm{E}$, respectively. The parameter $\delta_i\in \left \{ 0,1 \right \}$ is defined for convenience, in order to incorporate two situations: for $\delta_{\rm E}=0$, we have an EMI-aware eavesdropper which is capable to cancel out the interference, when it is intentionally originated from a colluding jammer. For $\delta_{\rm E}$=1, the EMI-unaware eavesdropper is affected by EMI just like the legitimate receiver. In all instances $\delta_{\rm B}=1$, implying that the legitimate receiver does not have the ability to cancel out interference. Now, for the subsequent analytical derivations, we need to find an approximate statistical distribution  
    for $X_i$. To that end, we resort to the Moment-Matching Method (MoM) as described in the following section.
\section{Approximate RIS Channel Modeling}
Although the MoM can be used for any target distribution, we formally demonstrate that Gamma and Exponential distributions offer excellent performance for approximating the equivalent RIS channels. Hence, in the following Propositions, we approximate $X_i$ as follows:
\begin{proposition}\label{Propos1}
The distribution of $X_\mathrm{B}$ can be approximated by a Gamma distribution, which is characterized by two
parameters $k_\mathrm{B}$ and $\theta_\mathrm{B}$, i.e.,
\begin{align}\label{eq7}
X_\mathrm{B}\sim\text{\rm Gamma}(k_\mathrm{B},\theta_\mathrm{B}), 
\end{align}
where
\vspace{-4mm}
\begin{align}\label{eq8}
k_\mathrm{B}=\frac{\varsigma_\mathrm{B}^2}{\iota_\mathrm{B}-\varsigma_\mathrm{B}^2}, \hspace{2mm} \theta_\mathrm{B}=\frac{\iota_\mathrm{B}-\varsigma_\mathrm{B}^2}{\varsigma_\mathrm{B}},
\end{align}
and
\vspace{-4mm}
\begin{align}\label{eq9}
\iota_\mathrm{B}=&\frac{  2\beta_{{\rm d},\mathrm{B}}^2M+4\beta_{{\rm d},\mathrm{B}}\omega\tr({\bf{\Theta}})+2|\omega\tr({\bf{\Theta}})|^2+2\omega^2\tr({\bf{\Theta}}^2)}{\tfrac{ A \sigma^2_{\mathrm{EMI}}}{\sigma_{\mathrm{B}}^2}\delta_{\rm B} \alpha_\mathrm{B}+2\delta_{\rm B} \alpha_\mathrm{B}+1 } \nonumber \\
\varsigma_\mathrm{B}=&\frac{ \beta_{{\rm d},\mathrm{B}}M+\omega \tr({\bf{\Theta}})}{\delta_{\rm B} \alpha_\mathrm{B}+1}, \hspace{2mm} \textit{for} \hspace{2mm} \omega=\beta_1 \beta_{2,\mathrm{B}}MA^2,
\end{align}
and ${\bf{\Theta}}={\bf{R}}^{\rm H} \bm{\Phi}{\bf{R}}{\bf{R}}^{\rm H}\bm{\Phi}^{\rm H}{\bf{R}}$, and $\alpha_i=\tfrac{ A^2 \sigma^2_{\mathrm{EMI}}}{\sigma_{i}^2} \beta_{2,i} \tr\left (  {\bf{R}}^{\rm H} \bm{\Phi}{\bf{R}}\bm{\Phi}^{\rm H}{\bf{R}}\right )$ for $i \in \left \{ \mathrm{B},\mathrm{E} \right \}$.
\end{proposition}
\begin{proof}
See Appendix~\ref{appendix1}.
\end{proof}
\begin{proposition}\label{coro1}
The distribution of $X_\mathrm{E}$ can be approximated as a Exponential distribution, i.e.,
\begin{align}\label{eq10}
X_\mathrm{E}\sim\text{\rm Exponential}(\theta_\mathrm{E}),
\end{align}
where
\vspace{-4mm}
\begin{align}\label{eq11}
\theta_\mathrm{E}=\frac{A^2\beta_1 \beta_{\rm{d},\mathrm{E}}\beta_{2,\mathrm{B}}M \tr\left ({\bf{R}}^{\rm H}\bm{\Phi}^{\rm H}{\bf{R}}{\bf{R}}^{\rm H}\bm{\Phi} {\bf{R}}\right ) +\epsilon+\xi+\varrho}{\left( \delta_{\rm E}\alpha_\mathrm{E}+1\right)\left(\beta_{{\rm d},\mathrm{B}}M+A^2\beta_1 \beta_{2,\mathrm{B}}M \tr\left ({\bf{\Theta}}\right ) \right)},
\end{align}
in which $\epsilon=A^4\beta_1^2 \beta_{2,\mathrm{B}}\beta_{2,\mathrm{E}}M \tr\left (\Upsilon \Upsilon  \right )$, $\xi=A^2\beta_1 \beta_{2,\mathrm{E}}\beta_{\rm{d},\mathrm{B}}M \tr\left (\Upsilon\right )$, and $\varrho=M\beta_{\rm{d},\mathrm{B}}\beta_{\rm{d},\mathrm{E}}$ for $\Upsilon={\bf{R}}^{\rm H} \bm{\Phi}{\bf{R}}{\bf{R}}^{\rm H}\bm{\Phi}^{\rm H}{\bf{R}}$.
\end{proposition}
\begin{proof}
See Appendix~\ref{appendix1}.
\end{proof}

\section{PLS Performance}
According to \eqref{eq6}, the required Eve's PDF and Bob's CDF to derive the PLS performance metrics are given below.
\subsection{SNR Distributions}
\vspace{-0.5mm}
\subsubsection{Distribution of $\gamma_\mathrm{E}$}
Using \eqref{eq10}, the PDF of the received SNR at $\mathrm{E}$ is obtained by carrying out a standard transformation of variables from \eqref{eq6}, i.e., $\gamma_\mathrm{E}=\overline\gamma_\mathrm{E}|X_\mathrm{E}|$. This yields,
\begin{equation}\label{eq12}
f_{\mathrm{E}}(\gamma_\mathrm{E})  =\frac{ 1}{\overline\gamma_\mathrm{E}\theta_\mathrm{E}}\exp\left(- \frac{\gamma_\mathrm{E}}{\overline\gamma_\mathrm{E}\theta_\mathrm{E}}  \right ).
\end{equation}
\subsubsection{Distribution of $\gamma_\mathrm{B}$}
Using \eqref{eq7}, and after a simple transformation of variables, i.e., $\gamma_\mathrm{B}=\overline\gamma_\mathrm{B}|X_\mathrm{B}|$, the CDF of the received SNR at B is obtained as
\begin{equation}\label{eq13}
F_{\mathrm{B}}(\gamma_\mathrm{B})  =\frac{\Upsilon\left( k_\mathrm{B}, \frac{\gamma_\mathrm{B}}{\overline{\gamma}_\mathrm{B}\theta_\mathrm{B}} \right)}{\Gamma\left( k_\mathrm{B} \right)}.
\vspace{-3.5mm}
\end{equation}
\subsection{SOP Analysis}
In this section, we consider the well-known wiretap PLS setup for passive eavesdropping (e.g., Eve only monitors the network by trying to intercept the messages), where the channel state information of Eve's channel is not available at Alice. In this scenario, Alice's only choice is to encode the data into codewords at a constant secrecy rate  $R_{\mathrm{S}}$. According to \cite{Barros}, the secrecy capacity is computed as $C_\mathrm{S}=\!\text{max}\left \{C_\mathrm{B}-C_\mathrm{E},0  \right \}$,
wherein $C_\mathrm{B}=\log_2(1+\gamma_\mathrm{B})$ and $C_\mathrm{E}=\log_2(1+\gamma_\mathrm{E})$
are the channel capacities at $\mathrm{B}$ and $\mathrm{E}$, respectively. Notice that secure communication can be guaranteed only in those instants when $R_\mathrm{S}\leq C_{\mathrm{S}}$, and is compromised otherwise, i.e., a secrecy outage occurs (e.g., data leakage to E). Mathematically, the $\text{SOP}=\Pr\left \{ C_\mathrm{S} < R_{\mathrm{S}}  \right \}$, and its tight lower bound $\text{SOP}_{\text{L}}$ can be defined as~\cite{Barros}
 \begin{align}\label{eq15}
\text{SOP}_{\text{L}}
 &=\int_{0}^{\infty}F_{\gamma_\mathrm{B}}\left ( 2^{R_{\mathrm{S}}} \gamma_\mathrm{E}\right )f_{\gamma_\mathrm{E}}(\gamma_\mathrm{E})d\gamma_\mathrm{E}.
\end{align}
\begin{proposition}\label{Propo2}
The $ \text{SOP}_{\text{L}}$ for RIS-assisted MISO wireless communications affected by EMI is obtained as
\begin{align}\label{eq17}
\text{SOP}_{\text{L}}=&\left ( \frac{\overline{\gamma}_\mathrm{E} \theta_\mathrm{E}}{\overline{\gamma}_\mathrm{B} \theta_\mathrm{B}} \right )^{k_\mathrm{B}}\frac{2^{k_\mathrm{B}R_{\mathrm{S}}}}{k_\mathrm{B} \mathcal{B}\left ( k_\mathrm{B},1 \right ) }  
\nonumber \\ & \times
{}_2F_1\left ( k_\mathrm{B}+1,k_\mathrm{B};1+k_\mathrm{B};- \frac{2^{R_{\mathrm{S}}}\overline{\gamma}_\mathrm{E} \theta_\mathrm{E}}{\overline{\gamma}_\mathrm{B} \theta_\mathrm{B}}\right ).
\end{align}
\end{proposition}
\begin{proof}
$\text{SOP}_{\text{L}}$ can be obtained directly from~\cite[Eq.~(7)]{Kong} with the respective substitutions and after some manipulations.
\end{proof}
Notice that the $\text{SOP}_{\text{L}}$ \textit{without} EMI is also given by \eqref{eq17} as a by-product, by setting $\sigma_{\rm EMI}^2=0$ and substituting $k_i$ and $\theta_i$ for their corresponding ones in \eqref{eq8} and \eqref{eq11}.

\section{Numerical results and discussions} \label{sect:numericals}
We now evaluate the impact of EMI on the system secure performance, with special focus on the cases with EMI-aware and EMI-unaware eavesdropper; these are referred to as EA and EU in the sequel, for the sake of compactness. The phase noise term $\varphi_n$ for the RIS is modeled as a zero-mean Von Mises random variable (RV) with shape parameter $\kappa$ \cite{Badiu}, where a smaller $\kappa$ implies a larger phase error.
For all figures, a carrier frequency of 3 GHz is used, so $\lambda=0.1$m, and $P/\sigma_{i}^2=124$ dB, which corresponds to transmitting 20 dBm over 1 MHZ of bandwidth with a 10 dB noise figure, so that $\sigma_{i}^2=-104$ dBm. For all SOP traces, we set $R_{\mathrm{S}}=1$ bps/Hz. For Fig.~\ref{figen}, only the indirect RIS channels are considered, whereas Fig.~\ref{fig5} includes the presence of both direct and indirect channels. For convenience of discussion, we define $\rho=\tfrac{P \beta_1}{\sigma^2_{\mathrm{EMI}}}$ as in \cite{andrea} as the ratio between the signal power and EMI power at each of the RIS elements. Monte Carlo (MC) simulations are provided to double-check the validity of the approximations made throughout the analysis.

\begin{figure*}[ht!]
\vspace{-5.5mm}
\centering
\psfrag{A}[Bc][Bc][0.5]{$\rho=5$ dB}
\psfrag{B}[Bc][Bc][0.5]{$\rho=-5$ dB}
\psfrag{C}[Bc][Bc][0.5]{no EMI}
\psfrag{D}[Bc][Bc][0.5]{$\rho=5$ dB}
\subfigure[Average SNR vs. $N$.]{\includegraphics[width=0.32\textwidth]{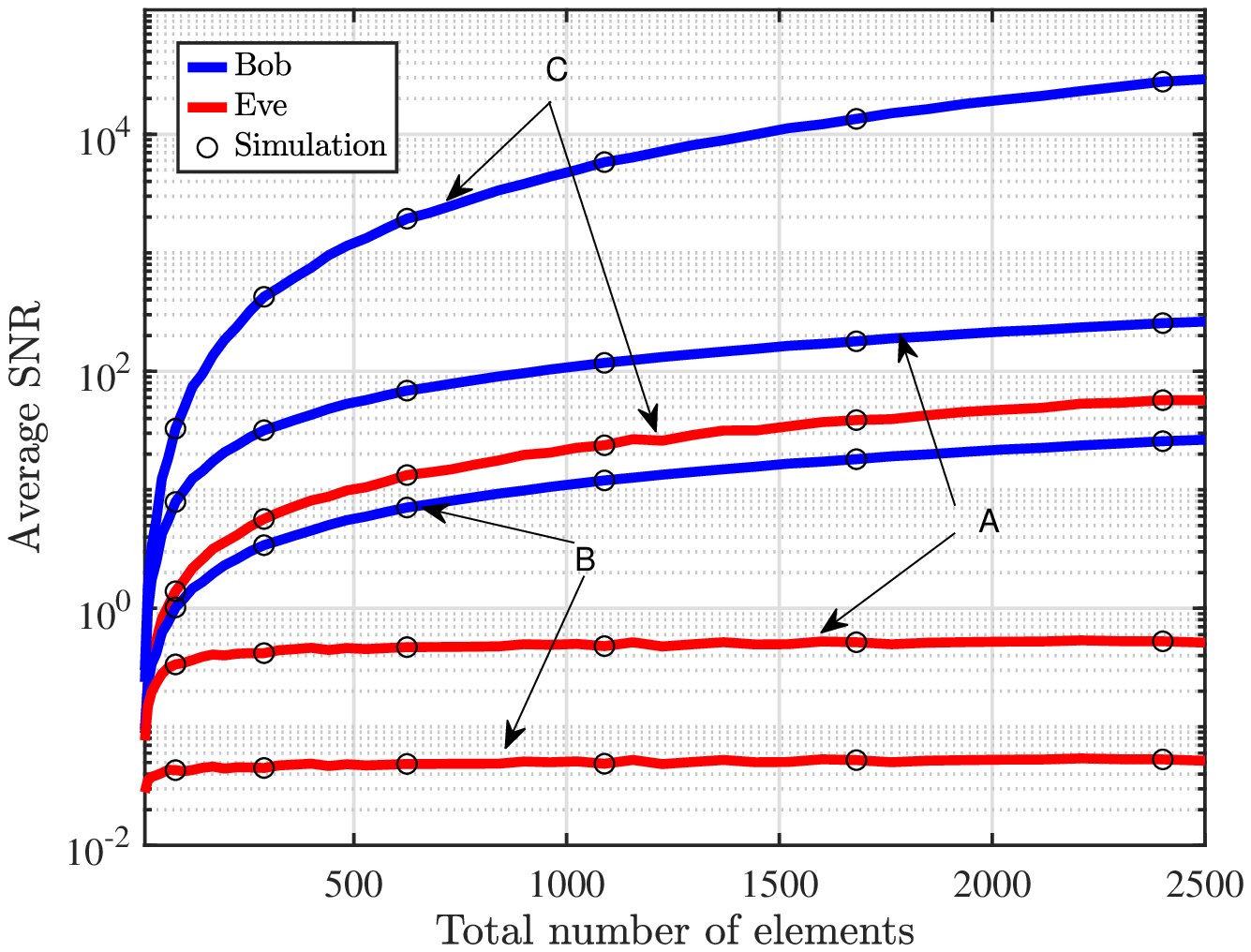}} 
\psfrag{E}[Bc][Bc][0.5]{$\rho=20$ dB}
\psfrag{A}[Bc][Bc][0.5]{$\rho=20$ dB}
\psfrag{B}[Bc][Bc][0.5]{$\mathrm{\textit{N}=324, 144}$}
\psfrag{C}[Bc][Bc][0.5]{$\sigma_{\text{EMI},\mathrm{B}}^2=-45$ dBm}
\psfrag{D}[Bc][Bc][0.5]{$\mathrm{\textit{N}=144, 324}$}
\subfigure[SOP vs. $\beta_{2,\mathrm{B}}$. ]{\includegraphics[width=0.32\textwidth]{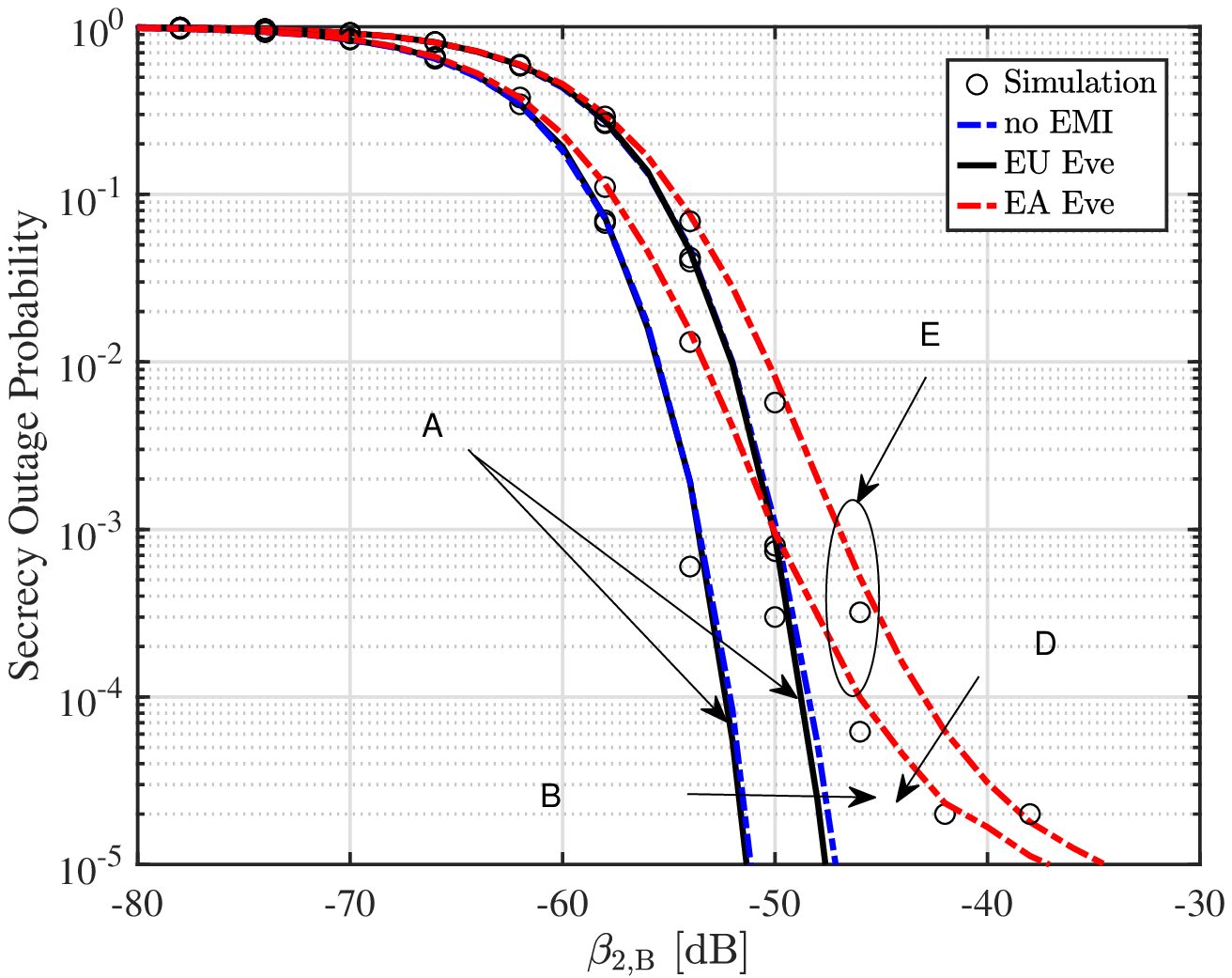}}
\psfrag{A}[Bc][Bc][0.5]{$\mathrm{\lambda/2}$}
\psfrag{B}[Bc][Bc][0.5]{$\mathrm{\lambda/5}$}
\psfrag{C}[Bc][Bc][0.5]{$\mathrm{\lambda/3}$}
\subfigure[SOP vs. $\rho$ for different inter-element spacing.]{\includegraphics[width=0.333\textwidth]{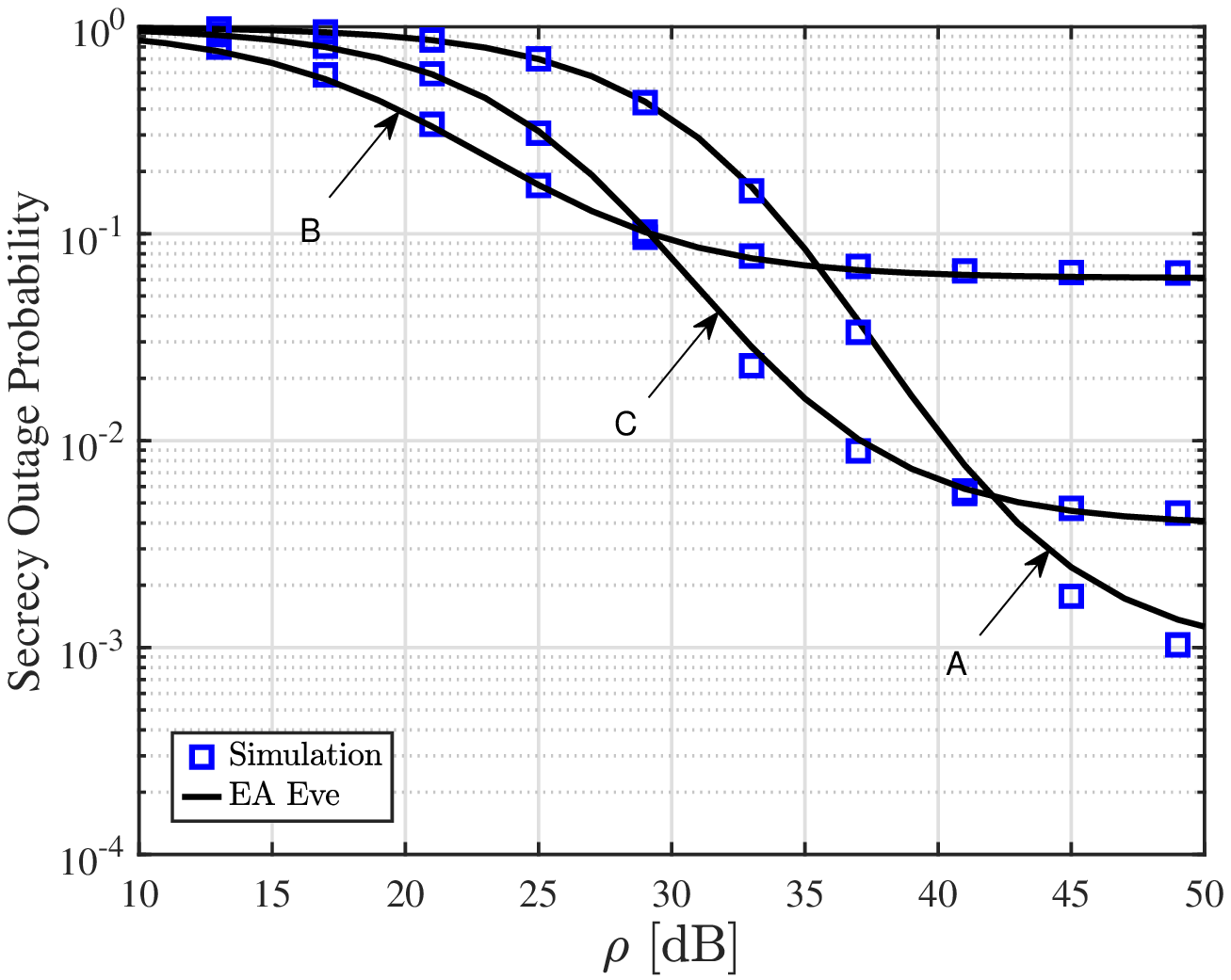}}
    \caption{SNR and SOP achieved for RIS-aided wiretap MISO system under EMI. In all figures, no direct link exists for $\mathrm{A}$-to-$\mathrm{B}$ or $\mathrm{A}$-to-$\mathrm{E}$. In~Figs.~\ref{figen}b-\ref{figen}c, the solid and dashed lines represent analytical solutions. Markers correspond to MC simulations.}
    \label{figen}
    \vspace{-4mm}
\end{figure*}

Fig.~\ref{figen}a illustrates the received average SNRs in the presence/absence of EMI, for both the legitimate and eavesdropper's links ($\mathrm{B}$ and $\mathrm{E}$, respectively) as a function of $N$. The system parameters are set to: $\kappa=3$, $A\beta_1=A \beta_{2,i}=-72$ dB, $M=4$, and the correlation matrix is built by setting $d_{\rm H}=d_{\rm V}=\lambda/4$. Here, we examine how the presence of EMI at the end nodes impacts the received SNRs. 
As pointed out in \cite{andrea} for the SISO case with perfect phase compensation, the average SNR at the legitimate receiver scales with $N^2$, i.e., it benefits from both the aperture and passive beamforming gains of the RIS. In the presence of EMI, the aperture gain (i.e., the ability of collecting energy proportionally to the size of the RIS) is cancelled out because it affects equally to the desired signal and the EMI; hence, Bob's SNR scales now with $N$. With regard to the case of Eve, she cannot benefit from any sort of passive beamforming gain, since the RIS is optimized only taking into account Bob's CSI. Thus, in the absence of EMI, Eve's SNR scales with $N$ (as stated in \cite{PLSRIS1}). When EMI affects the system, the aperture gain is also cancelled out and hence Eve's SNR tends to saturate (i.e., does not grow with $N$). We can also anticipate that in the EA scenario, since the scaling laws for the legitimate and eavesdropper's SNRs are coincident (i.e., scale with $N$), one of the key advantages of RIS for PLS is eliminated.

\begin{figure}[t]
\vspace{-2.5mm}
\centering
\psfrag{A}[Bc][Bc][0.5]{$\rho= 30, 25, 20$ dB}
\psfrag{B}[Bc][Bc][0.5]{$\rho=40$ dB}
 \includegraphics[width=65mm]{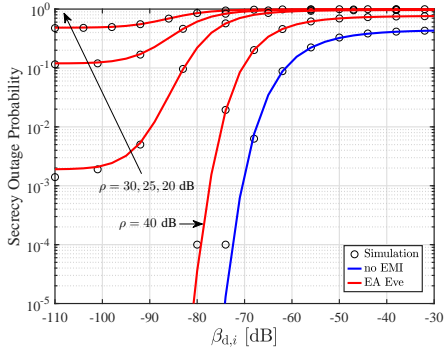}
\caption{SOP vs. $ \beta_{{\rm d},i}$ where both direct and indirect channels are present. Markers denote MC simulations whereas the solid lines represent analytical solutions.}
\label{fig5}
\vspace{-4.5mm}
\end{figure}

In Fig.~\ref{figen}b, we present the SOP vs. $\beta_{2,\mathrm{B}}$ for different values of $N$ and considering the EA/EU scenarios. The remaining system parameters are: $\kappa=7$, $A\beta_1=-47$ dB, $  A\beta_{2,\mathrm{E}}=-80$ dB, $M=5$,  $d_{\rm H}=d_{\rm V}=\lambda/3$, and $\rho=20$ dB. We see that the case without EMI and the EU scenario yield comparable secrecy performances, since the effect of EMI is similar for both the legitimate and illegitimate SNRs when the remaining parameters are fixed. However, in the EA scenario we see the effect of EMI severely deteriorates the secrecy performance, since Eve's ability to cancel out the interference causes the appearance of an irreducible SOP floor. This suggests the possibility of designing EMI-based attacks to the RIS with the help of a malicious jammer in collusion with Eve. In all instances, we note that the statistical approximations made in Section III closely match MC simulations.


Fig.~\ref{figen}c illustrates the effect of spatial correlation on the SOP as a function of $\rho$. The remaining parameters are set to: $N=100$ $\kappa=5$, $A\beta_1=A\beta_{2,i}=-58$ dB, and $M=2$. This figure explores the effects on the SOP of changing the area of the RIS by arranging the $N$ elements closer, which also increases spatial correlation, in the presence of EMI and assuming the EA scenario. We consider $d_{\rm H}=d_{\rm V} \in \left \{\lambda/2, \lambda/3,\lambda/5 \right \}$. Depending on the amount of EMI, two opposite behaviors are observed: in the high-EMI regime, i.e. lower values of $\rho$, a smaller RIS with closely arranged elements seems the preferred choice for having a better PLS performance. In this case, having a smaller RIS is beneficial since the amount of EMI collected by the RIS is lower. As $\rho$ grows, i.e. EMI power is reduced, and then having the RIS elements more separated starts becoming benefitial for PLS; in other words, the passive beamforming gain becomes dominant compared to the amount of EMI collected by the RIS. 

Finally,~Fig.~\ref{fig5} evaluates the SOP performance as a function of the direct path power $\beta_{{\rm d},i}$, i.e., the cases when both Bob and Eve experience a direct path. We assume the EA scenario with $\rho\in \left \{20, 25, 30, 40 \right \}$ dB, and include the case without EMI as a reference. The system parameter values are: $N=196$, $\kappa=5$, $A\beta_1=A \beta_{2,\mathrm{B}}=-52$ dB, $A\beta_{2,\mathrm{E}}=-62$ dB, $d_{\rm H}=d_{\rm V}=\lambda/4$, and $M=2$. Since in the EA scenario the EMI only affects the legitimate link, we see that the SOP improves as $\rho$ grows. We observe that as the LOS power for both links grows, (i.e., $\beta_{{\rm d},i}\uparrow$) the SOP tends to saturate to a large non-operational value due to EMI. For moderate values of $\rho$, a floor value for the SOP is also observed when $\beta_{{\rm d},i}\downarrow$, i.e., when the LOS power is not dominant. In this specific setup, assuming an operational SOP value of $10^{-3}$, this target PLS performance cannot be achieved when $\rho\in \left \{20, 25, 30 \right \}$ regardless of the LOS condition.
\vspace{-3mm}
\section{Conclusions}
The secrecy performance of RIS-assisted MISO wireless systems affected by EMI has been explored for the first time in the literature. As a first contribution, we provided simple but accurate approximations for the end-to-end RIS channels that incorporate key practical factors such as spatial correlation, phase-shift noise, the existence of direct and indirect channels, and EMI. We also discussed the effects of modifying the number and arrangement of the RIS elements on PLS performance, identifying the key relevance of Eve's EMI-awareness on system's secure performance. The sensitivity of RIS-assisted communications to EMI-based attacks with jammer-eavesdropper collusion is put forth, and the development of transmission and optimization techniques that improve PLS performance of RIS-assisted communications in the presence of EMI is identified as a challenging topic for future research.

\appendices

\section{Proof Propositions \ref{Propos1} and \ref{coro1}}
\label{appendix1} 
\subsection{Distribution of $X_\mathrm{B}$}
Let us match $X_\mathrm{B}$ to a Gamma distribution; hence, we have that its shape parameters are given by $k_\mathrm{B}=\tfrac{\left ( \mathbb{E}\left [ X_\mathrm{B} \right ] \right )^{2}}{\Var\left [ X_\mathrm{B} \right ]}$ and $\theta_\mathrm{B}=\tfrac{\left ( \Var\left [ X_\mathrm{B} \right ] \right ) }{\mathbb{E}\left [ X_\mathrm{B} \right ]}$. From \eqref{eq6}, let us define the RVs corresponding to the numerator and denominator as $Y_\mathrm{B}=| {{h_\mathrm{B}}}|^2 $ and $Z_\mathrm{B}=\tfrac{\delta_{\rm B} A \sigma^2_{\mathrm{EMI}}}{\sigma_{\mathrm{B}}^2} {\bf{h}}_{2,\mathrm{B}}^{\rm H}\bm{\Phi} {\bf{R}}{\bf{h}}_{2,\mathrm{B}} +1 $, respectively. Following a similar rationale as in \cite{Shah2000}, the 
desired and interfering signal terms in \eqref{eq1} are independent, provided that the interfering signal vector $\bm{\nu}$ is independent of the elements of the matrix ${\bf{G}}$ with circularly-symmetric Gaussian entries, and independent of the direct link  ${\bf{h}}_{{\rm{d}},\mathrm{B}}$. Hence, a lower bound of the mean value of $X_\mathrm{B}$ can be expressed as $\mathbb{E}[X_\mathrm{B}]=\mathbb{E}[Y_\mathrm{B}/Z_\mathrm{B}] \geq \mathbb{E}[Y_\mathrm{B}]/\mathbb{E}[Z_\mathrm{B}]$. From \eqref{eq2}, the mean value of the RV $Y_\mathrm{B}$ after some manipulations is
\vspace{-2mm}
\begin{align}\label{eqpropo1}
\mathbb{E}[Y_\mathrm{B}]=&\mathbb{E} 
\left [ \norm{ {\bf{h}}_{2,\mathrm{B}}^{\rm H}\bm{\Phi}{\bf{G}} + {\bf{h}}_{{\rm{d}},\mathrm{B}}^{\rm H}}^2\right ] \nonumber \\ \stackrel{(a)}{=}&\mathbb{E}\left [ \norm{ {\bf{h}}_{{\rm{d}},\mathrm{B}}^{\rm H} }^2 \right ]+\mathbb{E}\left [ \norm{{\bf{h}}_{2,\mathrm{B}}^{\rm H}\bm{\Phi}{\bf{G}}  }^2 \right ]\nonumber \\ =& \beta_{{\rm d},\mathrm{B}}M+A^2\beta_1 \beta_{2,\mathrm{B}}M \tr\left (  {\bf{R}}^{\rm H} \bm{\Phi}{\bf{R}}{\bf{R}}^{\rm H}\bm{\Phi}^{\rm H}{\bf{R}}\right ),
\end{align}
where step $a$, is given by the independence of the direct and indirect channels. 
The mean value of the RV $Z_\mathrm{B}$ is
\vspace{-2.5mm}
\begin{align}\label{eqpropo3}
\mathbb{E}[Z_\mathrm{B}]= &\mathbb{E}\left [ \tfrac{ \delta_{\rm B} A \sigma^2_{\mathrm{EMI}}}{\sigma_{\mathrm{B}}^2} {\bf{h}}_{2,\mathrm{B}}^{\rm H}\bm{\Phi} {\bf{R}}{\bf{h}}_{2,\mathrm{B}} +1\right ]
\nonumber \\ = &\tfrac{\delta_{\rm B} A^2 \sigma^2_{\mathrm{EMI}}}{\sigma_{\mathrm{B}}^2} \beta_{2,\mathrm{B}} \tr\left (  {\bf{R}}^{\rm H} \bm{\Phi}{\bf{R}}\bm{\Phi}^{\rm H}{\bf{R}}\right )+1.
\end{align}
Now, knowing that $\Var[X_\mathrm{B}]=\mathbb{E}[X_\mathrm{B}^2]-\mathbb{E}[X_\mathrm{B}]^2$, it is necessary to find $\mathbb{E}[X_\mathrm{B}^2]$. Hence, the second moment of the RV $Y_\mathrm{B}$ after some manipulations is
\begin{align}\label{eqpropo4}
\mathbb{E}[Y_\mathrm{B}^2]=&\mathbb{E} 
\left [ \norm{ {\bf{h}}_{2,\mathrm{B}}^{\rm H}\bm{\Phi}{\bf{G}} + {\bf{h}}_{{\rm{d}},\mathrm{B}}^{\rm H}}^4\right ] \nonumber \\ =& \mathbb{E}[  c ^2  ]+ 2\mathbb{E}[  d ^2  ]+ 2\mathbb{E}[ce] +\mathbb{E}[  e ^2 ],
\end{align}
where $c=||{\bf{h}}_{{\rm{d}},\mathrm{B}}^{\rm H}||^2$, $d=||{\bf{h}}_{{\rm{d}},\mathrm{B}}^{\rm H}{\bf{h}}_{2,\mathrm{B}}^{\rm H}\bm{\Phi}{\bf{G}}||$, and $e=||{\bf{h}}_{2,\mathrm{B}}^{\rm H}\bm{\Phi}{\bf{G}}||^2$. From~\cite[Eq.~(9)]{van}, we have that $\mathbb{E}[ c  ^2  ]=2\beta_{{\rm d},\mathrm{B}}^2M$. Also,
\vspace*{-2mm}
\begin{align}\label{eqpropo6}
\mathbb{E}[ d^2  ]=\beta_{{\rm d},\mathrm{B}}\beta_1 \beta_{2,\mathrm{B}}MA^2\tr\left (  {\bf{R}}^{\rm H} \bm{\Phi}{\bf{R}}{\bf{R}}^{\rm H}\bm{\Phi}^{\rm H}{\bf{R}}\right ).
\vspace*{-4mm}
\end{align}
From~\cite[Eqs.~(10,11)]{van}, we have that $\mathbb{E}[ d^2  ]=\mathbb{E}[ce]$ and
\begin{align}\label{eqpropo7}
\mathbb{E}[ e ^2  ]=&2|\beta_1 \beta_{2,\mathrm{B}}MA^2\tr\left (  {\bf{R}}^{\rm H} \bm{\Phi}{\bf{R}}{\bf{R}}^{\rm H}\bm{\Phi}^{\rm H}{\bf{R}}\right )|^2+
\nonumber \\ & 2(\beta_1 \beta_{2,\mathrm{B}}MA^2)^2 \tr \left ( \left ({\bf{R}}^{\rm H} \bm{\Phi}{\bf{R}}{\bf{R}}^{\rm H}\bm{\Phi}^{\rm H}{\bf{R}}  \right )^2 \right ).
\end{align}
 Then, the second moment of the RV $Z_\mathrm{B}$, is
\begin{align}\label{eqpropo9}
\mathbb{E}[Z_\mathrm{B}^2]= &\mathbb{E}\left [ \norm{ \tfrac{\delta_{\rm B} A \sigma^2_{\mathrm{EMI}}}{\sigma_{\mathrm{B}}^2} {\bf{h}}_{2,\mathrm{B}}^{\rm H}\bm{\Phi} {\bf{R}}{\bf{h}}_{2,\mathrm{B}} +1}^2\right ]=\mu^2 \eta+2\mu \eta+1,
\end{align}
in which $\mu=\tfrac{ \delta_{\rm B} A \sigma^2_{\mathrm{EMI}}}{\sigma_{\mathrm{B}}^2}$ and $\eta=A\beta_{2,\mathrm{B}} \tr\left (  {\bf{R}}^{\rm H} \bm{\Phi}{\bf{R}}\bm{\Phi}^{\rm H}{\bf{R}}\right )$. Finally, by combining \eqref{eqpropo1} to \eqref{eqpropo9}, the terms $k_\mathrm{B}$ and $\theta_\mathrm{B}$ can be attained
as in \eqref{eq8}. This completes the proof.
\vspace*{-2.5mm}
\subsection{Distribution of $X_\mathrm{E}$}
In this case, we claim that $X_\mathrm{E}$ can be approximated by an exponential RV with parameter $\theta_\mathrm{E}= \mathbb{E}\left [ X_\mathrm{E} \right ]$. Let us define the RVs $Y_\mathrm{E}=| {{h_\mathrm{E}}}|^2 $, and $Z_\mathrm{E}=\tfrac{\delta_{\rm E} A \sigma^2_{\mathrm{EMI}}}{\sigma_{\mathrm{E}}^2} {\bf{h}}_{2,\mathrm{E}}^{\rm H}\bm{\Phi} {\bf{R}}{\bf{h}}_{2,\mathrm{E}} +1 $. In the SISO case $Y_\mathrm{E}$ is exponentially distributed for $N\uparrow$ \cite{Badiu}, and independent of $Y_\mathrm{B}$ \cite{PLSRIS1}. This also applies to the MISO case with MRT beamforming based on Bob's CSI using the same rationale as in \cite{Shah2000}. A lower bound of the mean of $X_\mathrm{E}$ can be obtained as $\mathbb{E}[X_\mathrm{E}] \geq \mathbb{E}[Y_\mathrm{E}]/\mathbb{E}[Z_\mathrm{E}]$. The mean value of the RV $Z_\mathrm{E}$ is computed via \eqref{eqpropo3} with the respective substitutions and the expectation of the RV $Y_\mathrm{E}$ is given after some algebra by
\vspace{-2mm}
\begin{align}\label{eqpropo7}
\mathbb{E}[Y_\mathrm{E}]\geq& 
 \tfrac{\overset{T_1}{\overbrace{\mathbb{E}\left [ \norm{\left( {\bf{h}}_{2,\mathrm{E}}^{\rm H}\bm{\Phi}{\bf{G}} + {\bf{h}}_{{\rm{d}},\mathrm{E}}^{\rm H}\right)
 \left ( {{\bf{G}}^{\rm H}\bm{\Phi}^{\rm H}\bf{h}}_{2,\mathrm{B}} + {\bf{h}}_{{\rm{d}},\mathrm{B}}\right )}^2\right ]}}}{\mathbb{E}\left [\norm{ {\bf{h}}_{2,\mathrm{B}}^{\rm H}\bm{\Phi}{\bf{G}} + {\bf{h}}_{{\rm{d}},\mathrm{B}}^{\rm H}}^2 \right ]}  
\end{align}
where the numerator part in \eqref{eqpropo7} is calculated using \eqref{eqpropo1} and $T_1$ can be expressed as   
\vspace{-2mm}
\begin{align}\label{eqpropo8}
T_1 \stackrel{(b)}{=}&\mathbb{E}[ p ]+ \mathbb{E}[  q   ]+ \mathbb{E}[r] +\mathbb{E}[ s ],
\end{align}
in which step $b$, is given by the independence of Bob and Eve's channels \cite{PLSRIS1}, and $p=||{\bf{h}}_{{2},\mathrm{E}}^{\rm H}\bm{\Phi}{\bf{G}}{\bf{G}}^{\rm H}\bm{\Phi}^{\rm H}{\bf{h}}_{2,\mathrm{B}}||$, $q=||{\bf{h}}_{{\rm{d}},\mathrm{E}}^{\rm H} {\bf{G}}^{\rm H}  \bm{\Phi}^{\rm H}{\bf{h}}_{2,\mathrm{B}}||$, $r=||{\bf{h}}_{{2},\mathrm{E}}^{\rm H}\bm{\Phi}{\bf{G}}{\bf{h}}_{\rm{d},\mathrm{B}}||$, and $s=||{\bf{h}}_{{\rm{d}},\mathrm{E}}^{\rm H}{\bf{h}}_{{\rm{d}},\mathrm{B}}||$. Finally, it follows that, $\mathbb{E}[ p ]=A^4\beta_1^2 \beta_{2,\mathrm{B}}\beta_{2,\mathrm{E}}M \tr\left (\Upsilon \Upsilon  \right )$, $\mathbb{E}[  q   ]=A^2\beta_1 \beta_{\rm{d},\mathrm{E}}\beta_{2,\mathrm{B}}M \tr\left ({\bf{R}}^{\rm H}\bm{\Phi}^{\rm H}{\bf{R}}{\bf{R}}^{\rm H}\bm{\Phi} {\bf{R}}\right )$, $ \mathbb{E}[r]=A^2\beta_1 \beta_{2,\mathrm{E}}\beta_{\rm{d},\mathrm{B}}M \tr\left (\Upsilon\right )$, and $\mathbb{E}[ s ]=M\beta_{\rm{d},\mathrm{B}}\beta_{\rm{d},\mathrm{E}}$ for $\Upsilon={\bf{R}}^{\rm H} \bm{\Phi}{\bf{R}}{\bf{R}}^{\rm H}\bm{\Phi}^{\rm H}{\bf{R}}$.




\bibliographystyle{ieeetr}
\bibliography{bibfile}

\end{document}